\documentclass[10pt,aps,jmp,superscriptaddress]{revtex4-1}
\usepackage{amsfonts}

\usepackage{amsmath}
\usepackage{graphicx,color}
\usepackage{epstopdf}
\usepackage{epsfig}
\usepackage{mathrsfs}
\usepackage[breaklinks=true]{hyperref}

\newcommand{\beq}{\begin{equation}}
\newcommand{\eeq}{\end{equation}}
\newcommand{\bqa}{\begin{eqnarray}}
\newcommand{\eqa}{\end{eqnarray}}

\definecolor{green}{rgb}{0.00,0.50,0.00}

\newtheorem{theorem}{Theorem}

\newtheorem{corollary}[theorem]{Corollary}

\newtheorem{definition}[theorem]{Definition}

\newtheorem{lemma}[theorem]{Lemma}

\newtheorem{proposition}[theorem]{Proposition}
\newtheorem{remark}[theorem]{Remark}

\newenvironment{proof}[1][Proof]{\noindent\textbf{#1.} }{\ \rule{0.5em}{0.5em}}

\begin{document}

\title{Law of Large Numbers for Random Quantum Dynamical Semigroups}

\author{John E.~Gough} \email{jug@aber.ac.uk}
   \affiliation{Aberystwyth University, SY23 3BZ, Wales, United Kingdom}

\author{Yurii N. Orlov} \email{orlmath@keldysh.ru}
   \affiliation{Keldysh Institute of Applied Mathematics RAS}

	\author{Vsevolod Zh. Sakbaev} \email{fumi2003@mail.ru}
   \affiliation{Moscow Institute of Physics and Technology}
	
	\author{Oleg G. Smolyanov} 
   \affiliation{Lomonosov Moscow State University, Moscow Institute of Physics and Technology}
\date{\today}

\begin{abstract}
We present a Law of Large Numbers principle for uniformly continuous random quantum dynamical semigroups. Random iterates of independent copies of these semigroups are shown to be Chernoff equivalent to the quantum dynamical semigroup by the average generator.
\end{abstract}

\maketitle
 
Dedicated to the memory of our dear colleague Oleg Georgeivich Smolyanov (1938-2021).

\bigskip

\section{Introduction}

The aim of this paper is show that if $(\mathcal{L}_n )_n$ is an i.i.d. sequence of random (Lindblad) generators for quantum dynamical semigroups, each with mean $\overline{\mathcal{L}}$, then we may have a Law of Large Numbers principle
\begin{eqnarray}
e^{\frac{t}{n} \mathcal{L}_n} \circ \cdots \circ e^{\frac{t}{n} \mathcal{L}_1} 
\stackrel{\mathrm{LLN}}{\longrightarrow}
e^{t \overline{\mathcal{L}}}
\qquad \qquad (n \uparrow \infty)
 .
\end{eqnarray}

This type of principle has been established in a recent series of papers \cite{OSS14}-\cite{GOSS21}, for compositions of general random semigroups. The iterations may, in fact, be independent and identically distributed or at least just asymptotically so. An essential use is made of the Chernoff Theorem to handle the asymptotic convergence, and of Chebshev's inequality to control convergence. In the present paper, we extend this approach to quantum open systems, where the semigroups are uniformly continuous quantum dynamical semigroups and, in particular, the generators are of Lindblad form.

We will establish Chernoff equivalence of random iterates of independent copies of the semigroups with the semigroup generated by the mean Lindblad generator. The Chebshev inequality is established using an appropriate operator-algebra theoretic notion of the variance of random quantum dynamical maps.
We mention in connection with this physical situations where random Lindblad generators have recently appeared in the Physics literature, \cite{Ch,M}. 

Here, and in the following, we write $\mathbb{R}_+$ for the semi-axis $ [0,\infty )$ and we fix a Banach space $X$. We will be interested in 
one-parameter families of morphisms, $\phi = \{\phi_t (\cdot ) : t \in \mathbb{R}_+ \}$ on $X$, with $\phi_0 $ being the identity operator $id_X$ acting on $X$. We will sometimes write $\phi_t (x) $ as $\phi ( t ,x)$ for $t \ge 0$ and $x \in X$.The collection of all strongly continuous families of this type will be denoted by 
\begin{eqnarray}
Y = C_{\mathrm{str.}} ( \mathbb{R}_+, B(X)).
\end{eqnarray}
In particular, we will be interested in the subset $Y_0 $ of $C_0$-semigroups: these are the contraction semigroups in $Y$. With a standard abuse of notation, we will write $e^{t \mathcal{L}}$ for the semigroup with generator $\mathcal{L}$, even if the domain of the generator is only norm-dense in $X$.

\subsection{Chernoff Tangency}

One of the most important issues in the theory of dynamical systems is the relationship between dynamical maps and a possible (instantaneous) generator. In particular, we may say that two dynamical maps are $\phi, \psi \in Y$ instantaneously equivalent (say at time $t=0$) if we have $\phi_0=\psi_0$ and that their infinitesimal generators sufficiently coincide at $t=0$. For this purpose, we will use a notion of equivalence based on the well-known Chernoff Theorem.

\begin{theorem}[Chernoff]
Let $\phi : \mathbb{R}_+ \mapsto  B(X)$ with $ \| \phi_ t \| =1$,  $\phi _0 =  I   $, and possessing a strong derivative at $t=0$ on a dense domain $\mathcal{D} \subset X$
with $\mathcal{L} = \text{closure} \, \{\phi_0^\prime \vert_{\mathcal{D}} \}$ the generator of a contraction $C_0$-semigroup $\xi \in Y_0$. Then the sequence $\{\phi^{(n)} \}$ determined by $\phi^{(n)}_t = \phi_{t/n}^n$ converges strongly to $\xi_ t$ uniformly for  $t  $ in compact subsets of  $\mathbb {R}_+$  , that is, for every $x \in X$,
\begin{eqnarray}
\lim_{n\rightarrow \infty }\sup_{t\in \left[ 0,T\right] }\left\| \phi
_{t/n}^{n}\left( x\right) -\xi_t \left( x\right) \right\| =0.
\label{eq:Chernoff}
\end{eqnarray}
\end{theorem}
 
Here, $\phi_\tau^n$ means the $n$-fold composition of $\phi_\tau$. Using the observation that $\xi_t \equiv \xi _{t/n}^{n}$ for semigroups, we adopt the following definition of instantaneous equivalence (at $t=0$).

\begin{definition}
Let $\left( \phi _{t}\right) _{t\geq 0}$ and $(\psi _{t})_{t\geq 0}$ be elements of $Y = C_{\mathrm{str.}}  ( \mathbb{R}_+, B (X))$ then
we say that they are Chernoff equivalent (written $\phi \sim \psi$) if \cite{OSS14,OSS16}
\begin{eqnarray}
\lim_{n\rightarrow \infty }\sup_{t\in \left[ 0,T\right] }\left\| \phi
_{t/n}^{n}\left( x\right) -\psi _{t/n}^{n}\left( x\right) \right\| =0,
\end{eqnarray}
for all $x\in X$ and $T>0$.
\end{definition}

\begin{remark}
\begin{enumerate}
\item Chernoff equivalence gives an equivalence relation on $Y$.
\item The statement of the Chernoff Theorem may be rephrased with (\ref{eq:Chernoff}) replaced by the relation $\phi \sim ( e^{t\mathcal{L}} )_{t \ge 0}$.
\item Crucially, each equivalance class may have at most one element that is a semigroup.
\item If $\phi \sim \psi$, then both possess a strong derivative at $t=0$ which coincide on the essential domain of the generator of a semigroup in $Y$.
\end{enumerate}
\end{remark}

\begin{theorem}[\cite{Neklyudov08,OSS14}]
Let $\phi \in C_{\mathrm{str.}} (\mathbb{R}_+ , B( \mathfrak{H} ))$ with $\| \phi_t \|_{B(\mathfrak{H})} \leq e^{a t}$ for some real $a$, then the sequence $\{\phi^{(n)} \}$ determined by $\phi^{(n)}_t = \phi_{t/n}^n$ converges uniformly on compacts in the strong operator topology to a $C_0$-semigroup.
\end{theorem}


\subsection{Random Evolutions}
A \emph{random dynamics} $\Phi$ on the Banach space $X$ can be understood as a mapping $\Phi : \Omega \mapsto Y : \omega \mapsto \{  ( \phi_{t, \omega} )_{t \ge 0} \}$, where $\Omega $ is some sample space. Here we fix a probability space $(\Omega , \mathscr{A}, \mathbb{P})$. We require that $\Phi$ be measurable and this necessitates that we endow $Y$ it a suitable choice of Borel subsets. To this end, we note that $Y$ can be be given the topology $\tau$ generated by the family of semi-norms $\{ \rho_x,T : x\in X, T>0\}$ where $\rho_{x,T} (\phi) = \sup_{t \in [0,T]} \| \phi_t (x) \|$. We then take $\mathscr{B}$ to be the Borel subsets generated by the topology $\tau$. By a similar procedure, we may associate Borel subsets to any subset of $Y$. 

For instance, we denote by $\mathscr{B}_0$ the Borel subsets of the set of semigroups $Y_0$. A \emph{random semigroup} \cite{OSS16,Sakbaev18} may then be defined as a mapping $\Phi$ from $\Omega$ to $Y_0$ which is $(\mathscr{A}, \mathscr{B}_0)$-measurable. In more detail, a random semigroup is a mapping $\Phi : \mathbb{R}_+  \times X \times \Omega  \mapsto X $ such that $\Phi (t, x , \omega) $ is strongly continuous in $t$, linear in $x$ and measurable with respect to the $\sigma$-algebra $\mathcal{A}$, and we have the identity
\begin{eqnarray}
\Phi (t , \Phi (s, x , \omega ) , \omega ) = \Phi (t+s, x , \omega). 
\end{eqnarray}
We will write $\Phi (t, \omega )$ for the morphism $\Phi (t , \cdot , \omega )$ on $X$, in which case the identity reads as $\Phi (t , \omega ) \circ \Phi (s , \omega ) = \Phi (t+s , \omega )$.
We may write
\begin{eqnarray}
\Phi (t , \cdot ,\omega ) = e^{t \mathcal{L}( \cdot ,\omega )},
\end{eqnarray}
where $\mathcal{L} ( \cdot , \omega )$ is the generator.

\begin{definition} 
Let $\Phi$ be a random dynamics on the Banach space $X$ with underlying probability space $(\Omega , \mathscr{A}, \mathbb{P})$. The mean dynamics $\overline{\Phi}$ is understood as the Pettis integral
\begin{eqnarray}
\langle \sigma , \overline{\Phi}_t (x) \rangle = \int_\Omega \langle \sigma , \phi_{t, \omega } (x) \rangle \, \mathbb{P} [d \omega ],
\label{eq:Pettis}
\end{eqnarray}
for all $x\in X$ and all $\sigma$ in the pre-dual $X_\ast$, with $\langle \sigma, x \rangle$ is the duality pairing.
\end{definition}

One would hope that a random semigroup possesses a well-defined mean in $Y$ under suitable conditions. However, we should not expect the mean itself to form a semigroup. In fact, we will next recall results in this direction for Hilbert space dynamics.

\subsection{CP Maps}
In quantum theory, a central role is played by completely positive morphisms on the C*-algebra $\mathcal{B}$ of operators. To recall, given a map $\phi :\mathcal{B}\mapsto \mathcal{B}$ and a positive integer $n$ we define its
extension $\phi \otimes id_{n}$ to the algebra $\mathcal{B}\otimes M_{n}$,
where $M_{n}$ is the algebra of $n\times n$ matrices, as
\begin{eqnarray*}
\phi \otimes id_{n}\left( \left[ 
\begin{array}{ccc}
x_{11} & \cdots  & x_{1n} \\ 
\vdots  & \ddots  & \vdots  \\ 
x_{n1} & \cdots  & x_{nn}
\end{array}
\right] \right) = \left[ 
\begin{array}{ccc}
\phi (x_{11}) & \cdots  & \phi (x_{1n}) \\ 
\vdots  & \ddots  & \vdots  \\ 
\phi (x_{n1}) & \cdots  & \phi (x_{nn})
\end{array}
\right] .
\end{eqnarray*}
We say that $\phi $ is $n$-positive if $\phi \otimes id_{n}$ is positive. 
For instance, 2-positivity enforces the St\o rmer inequality $ \phi \left( x^{\ast }x\right) \geq \phi \left( x\right) ^{\ast }\phi \left(
x\right) $.

The map is completely positive if $\phi \otimes id_{n}$ is positive for all $%
n$. It is possible to give a unitary dilation for a CP map, that is we may find a second Banach space $\mathcal{A}$ such that
(for every $x \in \mathcal{B}$ and $\rho \in \mathcal{B}_\star$)
\begin{eqnarray}
\langle \rho, \phi (x) \rangle = 
\langle \rho \otimes \sigma , U^\ast ( x \otimes  I_{\mathcal{A}})U \rangle ,
\end{eqnarray} 
for $\sigma $ a positive normalized element of $\mathcal{A}$ and $U$ a unitary on $\mathcal{B} \otimes \mathcal{A} $. We will typically be interpreting the Banach spaces as subspaces of operators over Hilbert spaces, in which case $\mathcal{B} \otimes \mathcal{A} $ may be understood concretely as a tensor product. The physical interpretation is that $ \mathcal{A} $ is the environment, $\sigma $ is its state, and $U$ the evolution coupling the system $\mathcal{B}$ to $\mathcal{A}$.

We shall refer to one-parameter $C_0$-semigroups of CP maps as
quantum dynamical semigroups and denote the collection of such maps over $%
\mathcal{B}$ as $QDS\left( \mathcal{B}\right) $. 

The seminal result of
Lindblad \cite{Lindblad76} and Gorini-Kossakowski-Sudarshan \cite{GKS76} was the categorization of the
generators of uniformly continuous one-parameter semigroups of CP maps. 

\begin{theorem}[Lindblad Generators, \cite{Lindblad76,GKS76}]
\label{thm:LGKS}
The generator of a uniformly continuous QDS $\phi = (\phi_t )_{ t \ge 0}$ on $X= B (\mathfrak{h})$, where $\mathfrak{h}$ is a separable Hilbert space, takes the form
\begin{eqnarray}
\mathcal{L} (x) = \sum_k L^\ast_k x L_k + xK + K^\ast x,
\label{eq:GKLS_generator}
\end{eqnarray}
where $L_k, K \in B(X)$ and 
\begin{eqnarray}
 \sum_k L^\ast_k  L_k + K + K^\ast = 0,
\label{eq:generator_conservative}
\end{eqnarray}
where the sums in (\ref{eq:GKLS_generator}) and (\ref{eq:generator_conservative}) are understood to converge strongly.
\end{theorem}


\section{Random CP Dynamics}
We now fix a separable Hilbert space $\mathfrak{h}$ and set $X=B(\mathfrak{h})$ and consider one-parameter families of morphisms $(\phi_t)_{t \ge 0}$ belonging to $Y = C_{\mathrm{str.}} ( \mathbb{R}_+, B(X))$ with continuity understood as in the uniform topology. In particular, we will restrict our attentions to the subspaces $Y^{\mathrm{CP}}$ of CP families on $\mathcal{B} = B(\mathfrak{h})$ and $Y^{\mathrm{CP}}_0$ of CP $C_0$-semigroups. With these, we associate the Borel subsets $\mathscr{B}^{\mathrm{CP}}$ and $\mathscr{B}^{\mathrm{CP}}_0$, respectively.

\begin{definition}
A random quantum dynamical semigroup (random QDS) $\Phi =\left( \Phi _{t}\right) _{t\geq 0}$ is a measurable function taking values in $Y^{\mathrm{CP}}_0$. We take the underlying probability space $\left( \Omega ,\mathscr{A},\mathbb{P}\right) $ and so $\left( \Phi _{t}\right) _{t\geq 0}:\omega
\mapsto (\phi _{t,\omega })_{t\geq 0}$ is $\mathscr{A}$-measurable map taking values in $Y^{\mathrm{CP}}_0$. 
\end{definition}

Its average $\overline{\Phi}$ is then defined as above, though generally it will not form a semigroup. From Theorem \ref{thm:LGKS}, we have that, for each $\omega \in \Omega$, $(\phi_{\omega , t})_{t \ge 0}$ will have a generator $\mathcal{L}_\omega$ of the form (\ref{eq:GKLS_generator}). We refer to $:\omega \to \mathcal{L}_\omega$ as the random Lindblad generator associated with $\Phi$.

\begin{definition}
We say that a random QDS $\Phi$ is densely strongly equicontinuous if there exists a dense linear subspace $\mathcal{D} \subset X$ such that for every $\varepsilon >0$ there exists a $\delta > 0$ such that $\| \phi_{\omega, t_1} (x) - \phi_{\omega , t_2 } (x) \| < \varepsilon $, whenever $x \in \mathcal{D}, t_1, t_2 \ge 0$ and $ | t_1 -t_2 | < \delta$.
\end{definition}

Our first result is a straightforward specialization of an argument presented in \cite{OSS14,ES15}.

\begin{proposition}
Let $\Phi$ be a random uniformly continuous QDS on $\mathcal{B} = B(\mathfrak{h} )$. If $\Phi$ is uniformly bounded (that is, there exists an $M >0$ such that $\| \phi_{t, \omega }  \|_{ B(\mathfrak{H})} <M$ for all $t \ge 0$ and $\omega \in \Omega$) and densely strongly equicontinuous then $\overline{ \Phi}$ exists in $Y^{\mathrm{CP}}_0$.
\end{proposition}

The Chernoff Theorem provides sufficient conditions for the Chernoff equivalence of the average $\overline{\Phi}$ and the semigroup generated by the averaged generator $\overline{\mathcal{L}} = \mathbb{E} [ \Phi_t ^\prime ]$. The next result is adapted from \cite{OSS14,OSS16}.

\begin{proposition}
\label{prop:Lbar}
Let $\Phi$ be a random QDS with associated random generators $:\omega \mapsto \mathcal{L}_\omega$.  Further suppose that there exists an essential domain $D\subset \mathfrak{H}$ for the $\mathcal{L}_\omega$, with $\int_\Omega \| \mathcal{L}_\omega x \| \, \mathbb{P} [ d \omega ] < \infty$ for every $x\in \mathfrak{H}$. Then 
\begin{eqnarray}
\overline{\mathcal{L}} (x) = \int_\Omega \mathcal{L}_\omega (x) \, \mathbb{P} [d \omega ],
\end{eqnarray}
for all $x \in \mathcal{D}$, and defines an essentially self-adjoint operator (we denote its closure by the same symbol). Then the averaged QDS $\overline{\Phi}$ is Chernoff equivalent to the QDS with generator $\bar{\mathcal{L}} ( \cdot )$: that is, $\overline{\Phi} \sim ( e^{t \overline{\mathcal{L} }})_{t \ge 0}$.
\end{proposition}
\begin{proof}
It suffices to verify the conditions of the Chernoff Theorem. Since the functions $( \phi_{\omega , t} )_{t \ge 0}$ are continuous for each $\omega \in \Omega$ and takes values in the cone of identity-preserving CP maps on the Banach space $\mathcal{B} = B (\mathfrak{h})$ with $\Phi_{0, \omega } = id$, it follows that the average $\overline{\Phi}_t = \int_\Omega \phi_{t, \omega } \mathbb{P} [ d \omega ]$ inherits these properties.

For each $\omega\in \Omega$, we have the inequality $\phi_{t , \omega } = id + t \, \mathcal{L}_\omega + r_{t, \omega }$ where the remainder is bounded by $\| r_{t, \omega } \|_{\mathcal{B}} \le Ct^2 \, e^{\lambda t}$ where $\lambda = \sup_{\omega \in \Omega } \| \mathcal{L}_\omega \|_{\mathcal{B}}$.
This implies that $\psi (t) = id + t\, \overline{\mathcal{L}} +r_t$ satisfies the requirements of Chernoff's Theorem provided that $\| r_t \| \le Ct \, e^{\lambda t}$. The result then follows.
\end{proof}


\section{Law of Large Numbers for Random Quantum Dynamics}

\subsection{Compositions of Random Operators}
Let $\mathcal{L}$ be a random variable with values in the Banach space $X = B (B (\mathfrak{h}))$ of bounded linear operators acting in the C*-algebra $B (\mathfrak{h})$, defined as a weakly measurable mapping of the probability space $(\Omega, \mathscr{A}, \mathbb{P})$ into $X$. The weak measurability of the mapping $\mathcal{L}: \Omega \mapsto X: \omega \mapsto \mathcal{L} ( \omega )$ (i.e., the measurability of the collection of numerical functions $\langle g, \mathcal{L} (\omega )\, f \rangle$ for any $f\in X, g \in X^\ast$) in the case of a finite-dimensional space X is equivalent to the measurability of a mapping into a Banach space and implies the measurability of the real-valued random variable $ \| \mathcal{L}( \omega ) \|_X$.

Suppose that random variable $\mathcal{L}$ has mean value $\overline{\mathcal{L}}\in X$ and takes values in a ball of some radius $\lambda > 0$ of a Banach space $X$. Futhermore, suppose that $\mathcal{L}$ possesses a finite third moment: $\int_\Omega \|  \mathcal{L}(\omega )\|_X^3 \, \mathbb{P} [d \omega] <\infty $.

Let $( \mathcal{L}_n)_n $ be a sequence of independent identically distributed random variables, the distribution of each of which coincides with the distribution of the random variable $\mathcal{L}$.
Let $( \Phi_n) $ be a sequence of independent semigroups with each $\Phi_n$ generated by element $\mathcal{L}_n$, then the sequence $\Psi_n$ of random operator-valued functions defined by 
\begin{eqnarray}
\Psi_n(t)= \Phi_n(\frac{t}{n})\circ ...\circ \Phi_1(\frac{t}{n}),\ t\geq 0,\ n\in {\mathbb N},
\label{eq:Psi_n}
\end{eqnarray}
converges in probability to the one-parameter semigroup $( e^{\overline{\mathcal{L}}t})_{t \ge  0}$, in the space $C_s({\mathbb R}_+,X)$.

For each $t \ge 0$, the operator $e^{\mathcal{L}(\omega ) t}$ is then a random variable taking values in a ball of radius $\exp (t \lambda)$ in the space $X$. Consequently, the function $\overline{\Phi} :\ \mathbb{R}_+ \mapsto X$ defined by 
\begin{eqnarray}
\overline{\Phi} (t)= \int_{\Omega } e^{\mathcal{L}(\omega ) t}  \, \mathbb{P}(d\omega ). 
\end{eqnarray}
According to Theorem 1 \cite{OSS16}, the function $\overline{\Phi} $ is a map $\mathbb{ R}_+\to B(B(\mathfrak{h}))$ continuous in the strong operator topology (according to the finite dimensionality of the space $X$, the map is also continuous in the topology of the norm of the space $X$).

For each $\omega \in \Omega$, the Mean Value Theorem of Lagrange \cite{BogSm}, we have that  
\begin{eqnarray}
\| \exp (\mathcal{L} (\omega ) t)-{\bf I}\|_X \leq t\lambda \exp (\lambda t), 
\label{eq:Lagrange_estimate}
\end{eqnarray}
and $\| \exp (\mathcal{L} (\omega ) t)-{\bf I} -t  \mathcal{L} (\omega ) \|_X\leq \frac{t^2\lambda ^2}{2}\exp (\lambda t)$. From this estimate we have
$\frac{d}{dt}\overline{\Phi} (t)|_{t=0} =\overline{\mathcal{L}}$.

\begin{lemma}
\label{lem:D}
We have
\begin{eqnarray}
\mathbb{E} [ \Psi_n (t) ] = \overline{\Phi}  ( \frac{t}{n} )^n,
\end{eqnarray}
and if we introduce the operator-valued variance 
\begin{eqnarray}
\mathbb{D}_{\Psi_n} (t) = 
\mathbb{E}[(\Psi_n(t)-(\overline{\Phi} (\frac{t}{n}))^n)^*(\Psi_n(t)-(\overline{\Phi} (\frac{t}{n}))^n)],
\end{eqnarray}
then for each $T>0$ there exists a $C>0$ depending on $T$ and $\lambda$ such that $\sup_{[0,T]} \| \mathbb{D}_{\Psi_n} (t) \| \le \frac{C}{n}$.
\end{lemma}
\begin{proof}
For each $k$, we write $\Phi_k (\tau ) = \Phi^{(0)}_k (\tau )+ \Phi^{(1)}_k (\tau )$ where $\Phi^{(0)}_k (\tau ) = \mathbb{E} [\Phi_k (\tau ) ] = \overline{\Phi} (\tau )$. It follows that the random variables $\{ \Phi^{(1)}_k (\tau ) \}_k$ (the deviations from the mean) are mean zero independent variables. Furthermore, by virtue of (\ref
{eq:Lagrange_estimate}), we have almost surely, for all $t \ge 0$, the bound
\begin{eqnarray}
\| \Phi^{(1)}_k (\tau ) \|_{\mathcal{B}}  \le t \lambda \, e^{t \lambda }.
\label{eq:bound}
\end{eqnarray}
For each $n\in \mathbb  N$, we have the identity $(t  \ge 0 )$
\begin{eqnarray}
\Psi_n(t) = \bigg( \Phi_n^{(0)}(\frac{t}{n})+\Phi^{(1)}_n(\frac{t}{n}) \bigg) \circ ...\circ \bigg( \Phi^{(0)}_1(\frac{t}{n})+
\Phi_1^{(1)} (\frac{t}{n})\bigg) = \sum_{\boldsymbol{\alpha}} \Phi^{(\alpha_n )}_n (\frac{t}{n} ) \circ \cdots \circ \Phi^{(\alpha_1)}_1 (\frac{t}{n} )
\end{eqnarray}
where we have a sum over $\boldsymbol{\alpha} = (\alpha_n , \cdots, \alpha_1 ) \in \{0,1\}^n$. As the deviations $\{ \Phi^{(1)}_k (\tau ) \}_k$are mean zero independent we immediately obtain
\begin{eqnarray}
\mathbb{E} [ \Psi_n (t) ] = \Phi^{(0)}_n (\frac{t}{n} ) \circ \cdots \circ \Phi^{(0)}_1 (\frac{t}{n} )= \overline{\Phi}  ( \frac{t}{n} )^n .
\end{eqnarray}

The operator-valued variance may be written as
\begin{eqnarray}
\mathbb{D}_{\Psi_n} (t) = 
\mathbb{E}[\Psi_n(t)^* \circ \Psi_n(t)] - (\overline{\Phi} (\frac{t}{n})^\ast)^n \circ (\overline{\Phi} (\frac{t}{n}))^n
\end{eqnarray}
and we may similarly expand this as $\mathbb{D}_{\Psi_n} (t) = \sum_{\boldsymbol{\beta}, \boldsymbol{\alpha} } \mathbb{D}^{\boldsymbol{\beta}, \boldsymbol{\alpha} } (t) $ 
where
\begin{eqnarray}
 \mathbb{D}^{\boldsymbol{\beta}, \boldsymbol{\alpha} } (t) =\mathbb{E} \bigg[ 
\Phi^{(\beta_1 )}_1 (\frac{t}{n} )^\ast \circ \cdots \circ \Phi^{(\beta_n)}_n (\frac{t}{n} )^\ast \circ \Phi^{(\alpha_n )}_n (\frac{t}{n} ) \circ \cdots \circ \Phi^{(\alpha_1)}_1 (\frac{t}{n} )  \bigg] ,
\end{eqnarray}
with the exception of $\mathbb{D}^{\boldsymbol{0}, \boldsymbol{0} } (t)$ which will vanish identically.

We also have that $ \mathbb{D}^{\boldsymbol{\beta}, \boldsymbol{\alpha} } (t)$ will vanish identically if $\boldsymbol{\beta}\neq \boldsymbol{\alpha} $ since we otherwise will end up averaging a single $\Phi^{(1)}_k$ map which will yield a zero. Therefore we have $\mathbb{D}_{\Psi_n} (t) = \sum_{\boldsymbol{\alpha} } \mathbb{D}^{\boldsymbol{\alpha}, \boldsymbol{\alpha} } (t) $ and we may furthermore collect terms together as $\mathbb{D}_{\Psi_n} (t) = \sum_{m=0}^{n} \mathbb{D}_m^{n} (t)$ where $\mathbb{D}_m^{n} (t)$ is the contribution from the sum of those $\mathbb{D}^{\boldsymbol{\alpha}, \boldsymbol{\alpha} } (t) $ where there are exactly $m$ of the $n $ labels $\alpha_1, \cdots , \alpha_n$, take on the value 1. We have the bound $\| \mathbb{D}_m^{n} (t) \| \le \binom{n}{m}   \, e^{\lambda t} ( \lambda \frac{t}{n} )^m$ leading to the estimate
\begin{eqnarray}
\| \mathbb{D}_{\Psi_n} (t) \| 
\le \| \mathbb{D}^n_1 (t) \|  +
e^{\lambda t} \bigg[  (1+\frac{t}{n}\lambda )^n-1-t\lambda \bigg]
\end{eqnarray}
According to Taylor's Theorem, there is a number $s \in (0,1)$ such that 
\begin{eqnarray}
\| \mathbb{D}_{\Psi_n} (t) \| 
\le \| \mathbb{D}^n_1 (t) \|  + \frac{t \lambda^2}{2n} e^{\lambda t}(1 + s \frac{t}{n} )^n  
\le \| \mathbb{D}^n_1 (t) \|  + \frac{t \lambda^2}{2n} e^{2 \lambda t}   .
\end{eqnarray}
Since $\| \mathbb{E} [ \Phi_t ] \| _ {\mathcal{B}} \le e^{\lambda t}$ and by virtue of the bound (\ref{eq:bound}), we have $ \| \mathbb{D}^n_1 (t) \| \le n^2 e^{2 \lambda t} (\frac{t}{n} )^2$.

Therefore, for each $T >0$, there exists a $C= C(T, \lambda ) >0$ such that $\sup_{t \in [0,T]}  \| \mathbb{D}^{\Psi_n}  (t) \| < \frac{C}{n}$ for all $n \in \mathbb{N}$. 
\end{proof}

\bigskip

We may now use the Chebyshev inequality for operator valued measures, see Lemma 1 in \cite{OSS19}, to complete 
the result.


\begin{theorem}
\label{thm:main_LLN}
Let $\phi$ be a random quantum dynamical semigroup on $\mathcal{B} = B(\mathfrak{h})$ whose generators take values in ball of radius $\lambda < \infty$ in the Banach space $B( \mathcal{B})$. If $( \Phi _n ) _n $ is an independent sequence of random semigroups, each with the same distribution as $\Phi$, then the sequence $(\Psi_n )_n$ given by (\ref{eq:Psi_n}) satisfies the Law of Large Numbers
\begin{eqnarray}
\lim_{n \to \infty} \mathbb{P} [ \sup_{t \in [0,T]} \| \Psi_n (t) X - \overline{\Phi}_n (t) X \|_{\mathcal{B}} >\varepsilon ]=0,
\end{eqnarray}
for all $ X \in \mathcal {B}, T \ge 0$ and $\varepsilon >0$.
\end{theorem}

An immediate corollary is the following.

\begin{corollary}
Under the conditions of Theorem \ref{thm:main_LLN} is that the sequence $( \Psi_n )_n$ converges in probability to the quantum dynamical semigroup with average generator $\overline{\lambda} = \mathbb{E} [ \Phi_0^\prime ]$.
\end{corollary}

\section{Explicit Examples}

In the special case where $\mathfrak{h} = \mathbb{C}^N$ so that $\mathcal{B}$ is $M_{N}$, we have that the generator
takes the form
\begin{eqnarray}
\mathcal{L} \left( x\right) =\frac{1}{2}\sum_{n,m=1}^{N^{2}-1}k_{nm}\left( [a_{n}^{\ast
},x]a_{m}+a_{n}^{\ast }[x,a_{m}]\right) -i\left[ x,H\right] 
\label{eq:Kossakowski}
\end{eqnarray}
where $\left\{ a_{n}:n=1,\cdots ,N^{2}\right\} $, with $a_{N^{2}}=I$, is a
set of operators that are Hilbert-Schmidt orthonormal (that is, $tr\left\{
a_{n}^\ast a_{m}\right\} =\delta _{nm}$) forms a basis of $M_{N}$, $k=\left(
k_{nm}\right) $ is positive semidefinite and $H=H^{\ast }$ is Hermitean. The
matrix $k$ (called the Kossakowski matrix) and the operator $H$
(Hamiltonian) determine the generator.

In the case where $\mathfrak{h} \equiv \mathbb{C}^N$, then we have the straightforward result that if the generator $%
\mathcal{L}_{\omega }$ is determined from the Kossakowski matrix $k\left( \omega
\right) $ and Hamiltonian $H\left( \omega \right) $. The mean generator from Proposition \ref{prop:Lbar} is then $\bar{\mathcal{L} }$ will have
matrix $\bar{k}$ with entries $\int_{\Omega }k_{nm}\left( \omega \right)
\mathbb{P}[d\omega ]$ and $\bar{H}=\int_{\Omega }H\left( \omega \right) \mathbb{P}\left[
d\omega \right] $. Note that $\bar{k}$ will again be positive semidefinite and $H$ hermitean.


\end{document}